\documentclass{llncs}

\usepackage{microtype}
\usepackage{amsmath}
\usepackage{latexsym,amssymb}

\newif\iffull\fulltrue

\newif\ifdan\dantrue

\usepackage{txfonts}
\usepackage[all]{xy}
\usepackage{tikz}

\newcommand{\mi}{\mbox{\it min\,}}
\newcommand{\ma}{\mbox{\it max\,}}

\iffull

\else

\fi

\newcommand{\m}[2]{r_{#1#2}}

\usepackage{algorithm}
\usepackage[noend]{algorithmic}
\newcommand{\fun}[1]{{\textsl{#1}}}
\algsetup{linenodelimiter=.} 
\algsetup{indent=3em}

\ifdan\else
\usepackage{xfrac}
\fi

\title{Graph Theoretic Investigations \\ on Inefficiencies in Network Models
\vspace*{-.2cm}
}
\author{Pietro Cenciarelli, Daniele Gorla and Ivano Salvo
\vspace*{-.3cm}
}
\institute{Sapienza University of Rome, Dpt. of Computer Science\\
\texttt{\{cenciarelli,gorla,salvo\}}\texttt{\,@\,di.uniroma1.it}
\vspace*{-.5cm}
}

\begin{document}

\sloppy

\maketitle


\begin{abstract}
We consider network models where information items flow 
from a source to a sink node.
We start with a model where routing is constrained by energy available on nodes in finite supply 
(like in Smartdust) and efficiency is related to energy consumption. 
We characterize graph topologies ensuring that every saturating flow under every energy-to-node
assignment is maximum and  provide a polynomial-time algorithm for checking this property.
We then consider 
the standard flow networks with capacity on edges, where again efficiency is related to
maximality of saturating flows, and a traffic model
for selfish routing, where efficiency is related to latency at a Wardrop equilibrium. 
Finally, we show that all these forms of inefficiency 
yield different classes of graphs (apart from the
acyclic case, where the last two forms generate the same class).
Interestingly, in all cases 
inefficient graphs can be made efficient
by removing edges; this resembles a well-known phenomenon,
called {\em Braess's paradox}.
\vspace*{-.2cm}
\end{abstract}


\section{Introduction}

Through the years, several formal models have emerged
for studying network design in terms of network traffic, protocols, energy
consumption, and so on (see \cite{BI97,N80,S85,Iri98}, just to cite a few).
Our investigation of energy efficiency and load
balancing of multi-hop communication in ad-hoc networks started in \cite{cgs09},
where we considered a simple model, called \emph{depletable channels}, in which  
networks are oriented graphs with nodes equipped with a natural number
representing depletable charge, as in Smartdust~\cite{smartdust}. 


To better understand our model, consider a scenario in which four communication devices
are located in a landscape (Fig. \ref{fig:example}(a)). The devices may have different communication radii, e.g.
because of geographical reasons or because of their settings (a larger radius requires more 
energy to perform a communication). Thus, node reachability is not symmetric: in our example, $s$ can
reach (i.e., send information to) $u$ and $v$, but not vice versa. Such a scenario can be very
naturally modeled via a directed graph, where vertices correspond to communication devices and 
an edge is placed between $x$ and $y$ if $y$ falls within the communication circle centered in $x$.
In our example, the resulting graph is given in Fig. \ref{fig:example}(b).
To simplify reasoning, we will always assume that information flows from a single source (a device that does not 
receive information from anyone else -- $s$ in our example) to a single destination (a device that
does not send information to anyone else -- $t$ in our example).

Devices have a depletable amount of energy that is consumed throughout their life; moreover,
they are just information forwarders, so they can only send or receive information. 
Usually, a device consumes less energy when it receives than when it sends;
moreover, energy consumption is proportional to the length of the message exchanged.
This is modeled in our setting by slotting both information and energy, and by assuming that sending
one information unit consumes one energy unit. By contrast, receiving information does not lead to
any consumption. Furthermore, energy can only decrease during the life of the system; so, we  do not
model any form of recharge. This implies that every system has a lifetime: after sending a certain amount
of information from $s$ to $t$, all the intermediate nodes (or at least all those in a cut) will eventually die, due to lack
of energy. The resulting system is then called {\em dead} and the flow of information leading to it will be
called {\em inhibiting}.

\begin{figure}[t]
\begin{center}
\begin{tabular}{ccc}
\begin{minipage}{0.1\textwidth}
\qquad (a)
\end{minipage}
&
\begin{minipage}{0.5\textwidth}
\begin{tikzpicture}
\draw [color=red] (1,0) circle (2);
\draw [color=red] (1,0) node {$s$};

\draw [color= blue] (2.1,1.1) circle (1.3);
\draw [color= blue] (2.1,1.1) node {$u$};

\draw [color= olive] (2.8,0.3) circle (.8);
\draw [color= olive] (2.8,0.3) node {$v$};

\draw (3.1,0.8) circle (.4);
\draw (3.1,0.8) node {$t$};

\end{tikzpicture}
\end{minipage}
&
(b)
\begin{minipage}{0.4\textwidth}
$
\xymatrix@R=15pt@C=20pt{
& u \ar[dd]\ar[dr] & \\
s \ar[ru]\ar[rd] & & t \\ 
& v \ar[ru]
}
$
\end{minipage}
\end{tabular}
\caption{(a) Four communication devices with their communication radius: (b) the associated directed graph.}
\vspace*{-.8cm}
\label{fig:example}
\end{center}
\end{figure}
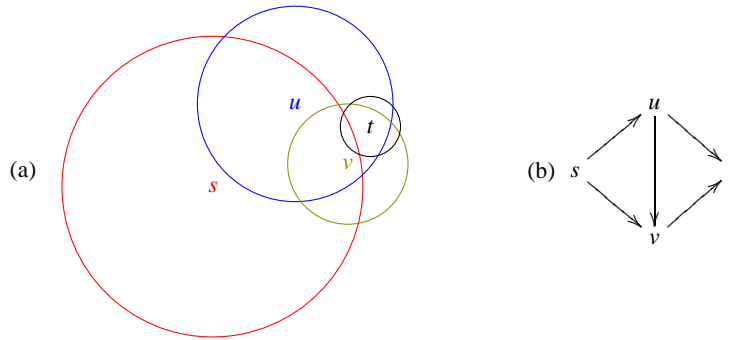

Aimed at capturing an abstract notion of communication service
provided by a network, in \cite{cgs09} we introduced an equivalence on networks
that equates two networks 
if and only if they have identical maximum and minimum inhibiting flow
 (there, we also proved that this corresponds to what
in the theory of concurrency is called {\em trace equivalence}).
Networks in which the minimum inhibiting flow is strictly less than the
maximum flow can be regarded as inefficient,
because an arbitrary routing may cause unnecessary depletion of energy.
As in \cite{Iri98}, we consider flows not controlled by a planning entity 
and thus any legal flow can take place.
We call {\em weak} those graph that, for some charge assignment to nodes,  
the minimum inhibiting flow is strictly less than the
maximum flow. For example, the channel whose topology is given by the graph 
in Fig.\ref{fig:example}(b) and where $u$ and $v$ have charge 1 (usually, we assume
that $s$ and $t$ have 'big' charge) can deliver at most 2 information items. However,
it is also possible that a single unit of information inhibits the channel, if it takes
the 'wrong' way (viz., $s\,u\,v\,t$). Thus, the graph in Fig.\ref{fig:example}(b) is weak.

As shown in this paper, network performance can be improved by
removing edges belonging to $st$-paths passing at least twice through a minimal vertex separator (that is
a minimal set of vertices whose removal disconnects the source from the sink).
In our example, removing the edge $u \rightarrow v$ leads to a non-weak graph.
This somehow resembles a counterintuitive but well-known situation in the
setting of selfish
routing in traffic networks,
called Braess's paradox \cite{braessFormal,braessOriginal}. 
Braess's paradox occurs when the
equilibrium cost may be
reduced by raising the cost of an edge or, equivalently, by removing
such an edge.
This form of inefficiency is known as {\em vulnerability}
\cite{Rough06,Milch06}.
Indeed,
the work presented here started with the empirical observation
that all natural examples of weak (acyclic) graphs turned out to be vulnerable. 
As this paper shows, this is not incidental. 

Both depletable channels and traffic networks have strong similarities
with the standard model of flow networks \cite{amo93}. In particular,
depletable channels
can be easily translated into standard (edge-capacitated) flow
networks. Usually, this
translation is painless, as it preserves all the standard notions in
flow networks, such as
maximum and minimum inhibiting flow. Interestingly,
if we define as {\em edge-weak} those graphs that, for some assignment of capacities to edges, 
admit a non-maximum inhibiting flow, it turns
out that weakness and edge-weakness {\em do not} coincide.

In this paper, we compare these three kinds of inefficiency 
from a graph-theoretical perspective. 
Our main contributions are:

1. In Section \ref{sec:weak}, we characterize {\em weak graphs}
as those graphs with an $st$-path that passes twice through the same minimal vertex separator.

2. In Section \ref{sec:complexity}, we then show that weakness can be checked in polynomial time.

3. In Section \ref{sec:edgeweak}, we characterize {\em edge-weak graphs}
as those graphs with an $st$-path that passes twice through the same cut-set 
(i.e., the set of edges that connect the two blocks of a cut), for a proper kind of cut that
we call {\em connected}.

4. In Section \ref{sec:vuln} we characterize {\em vulnerable graphs} 
as those graphs containing a subgraph homeomorphic to Fig.\ref{fig:example}(b).
This extends previous results of \cite{Milch06,ChenEtal15} to arbitrary directed graphs and 
completely solves a
question left open in \cite{Rough06}. 
Indeed, in \cite{Milch06} the characterization holds only for undirected graphs, whereas
in \cite{ChenEtal15} the result holds only for a specific
class of directed graphs (that they call {\em irredundant}).

5. Finally, in Section \ref{sec:compare}, we relate all the 
above mentioned class of graphs. 
In the general case, vulnerability implies edge-weakness. 
Moreover, vulnerable graphs always contain 
an acyclic weak subgraph.
This suggests that the core reason that makes a graph vulnerable is indeed its
weak subgraphs.
If we restrict our attention to DAGs, 
weakness implies both
vulnerability and edge-weakness. 
These two classes coincide and they coincide with the class of graphs 
that are not series-parallel \cite{RS42}.
%

Due to space constraints, all proofs are relegated to the Appendix.

\vspace{-2mm}
\section{Inefficiency in Depletable Channels: Weakness}
\label{sec:weak}
\vspace{-1mm}

A {\em depletable channel} is a graph $G$ equipped
with a function $\eta$ associating with each node a non-negative integer representing its
\emph{depletable charge}. 
Throughout the paper, we only consider \emph{directed simple st-graphs}, that are  
directed graphs without self-loops and parallel edges, with 
a fixed source vertex $s$ and sink vertex $t$. 
When no confusion arises,
we write just $\eta$ for a channel $(G, \eta)$ and call $G$ its \emph{underlying
graph}. 

Charges may change as result of information passing through the net.
Each item passing through a node consumes one unit of the node's charge, thus
leaving the channel in a state of lower energy.
We also assume
that the charge of the source and of the sink are always 
large enough to be irrelevant 
(say, $\infty$ for simplicity). 

An \emph{st-path} (just path from now on) in a 
graph is a (possibly cyclic) directed walk from $s$ to $t$. 
The set $P(\eta)$ is the set of all paths of $\eta$. 
Since we only consider source-to-sink paths, we can assume that every vertex belongs
at least to one of such paths.
%
%
A \emph{flow} for $\eta$ is a function 
$\phi:P(\eta) \rightarrow \mathbb{N}$ such that
$\phi(v) \leq \eta(v)$, for every $v \in V$.
Here, $\phi(v)$ denotes the amount of $v$'s charge consumed by $\phi$,
i.e., $\phi(v) = \sum_{p \in P(\eta)} \m v p \cdot \phi(p)$,
where $\m v p$ is the number of times in which node
$v$ is repeated in the path $p$ (0, if $v\not\in p$).
With respect to standard flow models, 
we admit a positive flow over cycles because in our model each node
knows its neighborhood only. Without global information about the net,     
routing can easily lead to cyclic paths. 

The {\em value} of $\phi$ is $\sum_{p \in P(\eta)} \phi(p)$.
We denote by $\ma_\eta$ the maximum value of a flow for $\eta$.
We call $\eta$ a \emph{dead} channel if $\ma_\eta=0$. 
The residual of a channel $(G,\eta)$ after a flow $\phi$ is a channel with
the same underlying graph $G$ and with a charge function $\eta'$ such that
$\eta'(v) = \eta(v) - \phi(v)$, for every vertex $v$.
A flow $\phi$ is said to \emph{inhibit} $\eta$ if the residual of $\eta$ after $\phi$ is
dead. We denote by $\mi_\eta$ the smallest value of an inhibiting flow in $\eta$.

The form of inefficiency in this network model is given by the following definition.

\begin{definition}
	\vspace*{-1mm}
\label{def:weak}
A graph $G$ is \emph{weak} if $\mi_\eta\not=\ma_\eta$, for some
channel $\eta$ whose underlying graph is $G$.
	\vspace*{-2mm}
\end{definition}

Two typical examples of channels are depicted in Fig.~\ref{rombi}, where we depict
the charge of a node in place of its name and use '$\circ$' to denote a large-enough charge.
Here, channel (1) is not weak, since every routing of messages always ensures the delivery
of 2 information units (remember that every information unit consumes one charge unit of
every traversed node). By contrast, channel (2) is weak because 
there is also a inhibiting flow that only delivers
1 unit (by sending the information unit along the path $\circ\,2\,1\,1\,2\,\circ$).

\begin{figure}[t]
\begin{center}
\begin{tabular}{c@{\hspace{.5cm}}l@{\hspace{1.5cm}}c@{\hspace{.5cm}}l}
(1) & \raisebox{.5cm}{\xymatrix@R=4pt@C=16pt{&& 1 \ar[dr] \\
   \circ \ar[r] & 2 \ar[ur] \ar[dr] & & 2 \ar[r] & \circ \\
   & & 1 \ar[ur]}}
&
(2) & \raisebox{.5cm}{\xymatrix@R=4pt@C=20pt{&& 1 \ar[dr] \ar[dd]\\
   \circ \ar[r] & 2 \ar[ur] \ar[dr] & & 2 \ar[r] & \circ \\
   && 1 \ar[ur]}}
\vspace*{-.4cm}
\end{tabular}
\end{center}
\caption{Two channels}
\label{rombi}
\vspace*{-2mm}
\end{figure}
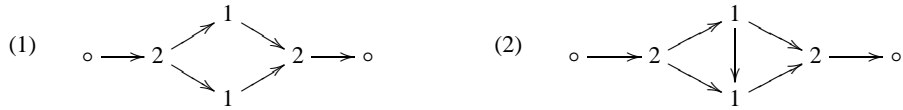

%


Weakness can be characterized in terms of the existence of a walk that 
passes at least twice through a 
{\em minimal st-separator} (or {\em mvs}, for minimal vertex separator, 
assuming $s$ and $t$ fixed). An mvs \cite{Golumbic80} is a 
minimal set of vertices whose removal disconnects $s$ and $t$.

\begin{theorem}
		\vspace*{-1mm}
\label{thm:weak-characterize}
A graph is weak if and only if there exists an mvs $T$ and a directed walk
$a\leadsto b$ with $a,b\in T$.
	\vspace*{-2mm}
\end{theorem}

One direction is proved by assigning charge 1 to all nodes in $T$ and a big value
to the remaining ones; we then
exploit $a\leadsto b$ to saturate $T$ with a flow
of value smaller than 
$\ma_\eta$.
For the converse, the existence of an inhibiting flow of value smaller than
the maximum implies existence of an mvs saturated by this flow; this can only 
happen if there is a path passing twice through the mvs.

\vspace{-3mm}
\section{Checking Weakness}
\label{sec:complexity}
\vspace{-2mm}
In \cite{cgs09}, we proved that, given a graph
and a charge $\eta$ to its nodes, it is NP-hard to determine whether $\min_\eta \neq \max_\eta$. 
By contrast, we show in the following that checking weakness is a polynomial-time problem.

Stemming 
from Theorem~\ref{thm:weak-characterize}, a trivial algorithm to determine if a graph is weak 
is to generate all mvs's and, for each of them, check if there exists a walk that touches the mvs twice. 
Unfortunately, the number of mvs's in a graph can be 
exponential in $|V|$ \cite{KK98}. 
However, as we now show, it is enough to examine at most $|V|^2$ mvs's.

From now on, we will use the following terminology, inspired by Theorem~\ref{thm:weak-characterize}. 
A walk $a\leadsto b$ is {\em critical} if there exists an mvs $T$ such that $a,b\in T$. 
In such a case, we call $b$ a {\em critical node} and $T$ a {\em critical} (or even {\em b-critical}) mvs.
%

\begin{figure}[t]
	\begin{tabular}{cc}
		\begin{minipage}{0.46\textwidth}
			\center
			$
			\xymatrix@R=12pt@C=16pt{
				& a\ar[rr]\ar[dr] & & u\ar[dr] & \\
				s \ar[ur]\ar[dr] & & b\ar[ur] & & t\\
				& p\ar[r]\ar[ur] & c\ar[urr] & &\\
			}
			$
			\caption{A graph with a complete chain not containing a specific critical node}
			\label{fig:nobcritical}   
		\end{minipage}
		~
		&
		~
		\begin{minipage}{0.46\textwidth}
			\center
			$
			\xymatrix@R=12pt@C=16pt{
				& & a\ar[r]\ar[d] & a_2\ar[dr] & \\
				s \ar[urr]\ar[dr] & & a_1\ar[d] & & t\\
				& p\ar[r] & b\ar[urr] & &\\
			}
			$
			\vspace{-3mm}
			\caption{A graph with a complete chain not containing any critical mvs}
			\label{fig:compl-chain}   
		\end{minipage}
	\end{tabular}
	\vspace{-5mm}
\end{figure}

We say that a node $u$ is {\em covered} by a set of nodes $A$, notation $u\sqsubseteq A$, 
if all walks starting from $u$ to the sink $t$
contain at least a node $v\in A$. 
A set of nodes $A$ is {\em covered} by a set of nodes $A'$, notation $A\sqsubseteq A'$, 
if $u\sqsubseteq A'$, for all $u\in A$.
A set of nodes $A$ {\em precedes} a node $u$, notation $A\preceq u$, 
if all walks starting from the source $s$ to 
the node $u$ contain at least a node $v\in A$. 
A set of nodes $A$ {\em precedes} a set of nodes $A'$, notation $A\preceq A'$,
if $A\preceq u$, for all $u \in A'$.
Since mvs's are minimal sets of $st$-separators, the following claim can be easily proved.

\begin{lemma}
\vspace{-1mm}
\label{fact:one}
If $A\preceq a$ or $a\sqsubseteq A$, then $A\cup\{a\}\not\subseteq T$, for every mvs $T$.
\vspace{-2mm}
\end{lemma}

Moreover, it is well-known \cite{F72} that the set of mvs's of a graph is partially ordered 
w.r.t. $\sqsubseteq$, with minimum element $\{s\}$ and maximum element $\{t\}$.
A sequence of mvs's $T_0, T_1, \ldots, T_n$ is a {\em chain} if, for all $i$ ($0\leq i<n$), we 
have $T_i\sqsubset T_{i+1}$; it is {\em complete} if $T_0=\{s\}$, $T_n=\{t\}$ and 
$T_i\sqsubseteq T \sqsubseteq T_{i+1}$ 
implies $T=T_i$ or $T=T_{i+1}$,
for every mvs $T$.

From \cite{SL97}, it follows that a complete chain contains at most $|V|$ mvs's. 
Following the idea of finding a critical mvs and a critical walk by examining a complete chain of mvs's, 
we first observe that a specific critical node may never appear in a complete 
chain. To see this, consider the graph depicted in Fig.~\ref{fig:nobcritical}. 
In the mvs $\{a,b,c\}$, $b$ is a critical node, but it does not appear 
in the complete chain 
$\{s\}, \{a,p\}, \{u,p\}, \{u,c\}, \{t\}$ that,
however, contains the critical node $u$. This is not incidental, 
as the following theorem states.

\begin{theorem}
	\vspace{-1mm}
\label{thm:criticalNodes}
If the graph is weak, then   
in every complete chain of mvs's $T_0, T_1, \ldots, T_n$ there exists at least 
a $T_i$ that contains a critical node.
	\vspace{-2mm}
\end{theorem}

Unfortunately, Theorem~\ref{thm:criticalNodes} is not enough to conclude that 
any complete chain of mvs's in a weak graph contains at least a critical mvs. 
To see this, let us consider the graph of Fig.~\ref{fig:compl-chain}:
it is weak, as the critical mvs $\{a,b\}$ testifies. However, if we consider the complete chain 
$C'=\{s\}, \{a,p\}, \{a_1,a_2, p\}, \{a_2,b\}, \{t\}$, we cannot find any critical mvs;
$C'$ just contains the non-critical mvs $\{a_2,b\}$ that contains the critical node $b$. 

Theorem~\ref{thm:bminimal} ensures that we can check if any node $b$ is critical 
(and hence check the existence of a $b$-critical mvs) 
by considering any {\em $b$-minimal mvs}, i.e.,
an mvs $T$ such that $b \in T$ and $b\not\in T'$, for every mvs $T'\sqsubset T$.

\begin{theorem}
		\vspace{-1mm}		
\label{thm:bminimal}
If $b$ is a critical node, then all $b$-minimal mvs's are $b$-critical.
\vspace{-2mm}
\end{theorem}

In the example of Fig.~\ref{fig:compl-chain}, it suffices to consider 
the mvs $\{a_2,b\}$: its $b$-minimal predecessor (w.r.t. $\sqsubset$) is $\{a,b\}$, that is
$b$-critical. 

Theorems~\ref{thm:criticalNodes} and \ref{thm:bminimal} are the main ingredients of 
our polynomial algorithm for checking graph weakness (function \fun{weakOrNotWeak} in 
Alg.~\ref{alg:wonw2}). 
Our algorithm generates a complete chain of mvs's, that contains at most 
$|V|$ mvs's. By Theorem~\ref{thm:criticalNodes}, we know that the graph
is weak if and only if at some point an mvs with a critical node appears in the chain.
Thus, for every new node $b$ appeared in all such mvs's, it computes 
a $b$-minimal mvs (by backward generating another complete chain of at most $|V|$ mvs's)
and it checks if it is $b$-critical. 
Thanks to Theorem~\ref{thm:bminimal}, this suffices to conclude.
\begin{algorithm}[t]
  \caption[weak or not weak]
  {Checking Weakness}
  \label{alg:wonw2}
  \begin{algorithmic}[1]
    \REQUIRE
    {A directed $st$--graph $G = (V,E,s,t)$, $s$ is the source and $t$ is the sink}
    \ENSURE {\fun{weakOrNotWeak}$(G)$} 
   \STATE $T\gets \{s\}$; 
	\WHILE{$T\not=\{t\}$}
    \STATE $T' \gets\fun{immediateMvsRight}(T)$
	\IF{$T'$ is critical}
	\STATE {{\bf return} {\sc Weak}}
    \ENDIF
	\FORALL{$b\in T'\setminus T$}
	\STATE $T^\ast=\fun{minimalMvs}(T', b)$
	\IF{$T^\ast$ is critical}
	\STATE {{\bf return} {\sc Weak}}
    \ENDIF
    \ENDFOR
    \STATE $T\gets T'$
    \ENDWHILE
	\STATE {\bf return} {\sc not Weak}
\end{algorithmic}
\end{algorithm}  

\begin{algorithm}[t]
  \caption[compute minimal mvs]
  {A $b$-minimal mvs smaller (w.r.t. $\sqsubset$) than $T$}
  \label{alg:mmvs}
  \begin{algorithmic}[1]
    \REQUIRE
    {An mvs $T$, a node $b\in T$}
    \ENSURE {\fun{minimalMvs}$(T,b)$} 
   \STATE $A\gets \{ u\in T~|~(T \cup \fun{pred}(u)) \setminus \{b\} \not\preceq b \}$
	\WHILE{$A\not=\varnothing$}
    \STATE \fun{choose} $u\in A$
    \STATE $T\gets T_u$
   \STATE $A\gets \{ u\in T~|~(T \cup \fun{pred}(u)) \setminus \{b\} \not \preceq b \}$
    \ENDWHILE
   	\STATE {\bf return} $T$
\end{algorithmic}
\end{algorithm}  

Function \fun{immediateMvsRight} 
adapts the work in \cite{SL97} to build a complete chain of mvs's, 
generated from the bottom mvs $T=\{s\}$ by iteratively proceeding as follows. 
Given an mvs $T$, pick up any vertex $u\in T$
and replace it with its immediate successors, denoted by $\fun{succ}(u) = \{v \in V : u \rightarrow v \in E\}$. 
The set $S=(T\setminus\{u\})\cup\fun{succ}(u)$ is a
separator, though not necessarily minimal: some vertices could be covered by other ones. 
Thus, we consider 
$T^u = S \setminus{\cal I}_t(S)$, 
where 
${\cal I}_t(S)$ 
contains all the vertices $v$ such that $v \sqsubset 
S \setminus \{v\}$; in particular, we always have $u \in {\cal I}_t(S)$.
As shown in \cite{SL97}, $T^u$ is an mvs and $T \sqsubset T^u$. 

As an example, let us consider the mvs $T=\{a,p\}$ in Fig.~\ref{fig:nobcritical}: 
$(T\setminus\{a\})\cup\fun{succ}(a)=\{b,u,p\}$ is 
not minimal. 
$T^a$ is $\{u,p\}$, since the node $b$ is covered by $u$.  

The mvs $T^u$ is not necessarily an immediate successor (w.r.t. $\sqsubset$) of $T$.
To obtain an immediate successor of $T$ (and thus build a complete chain), it suffices to consider 
the mvs $T'=\min_{u\in T}T^u$, where the minimum is calculated w.r.t. $\sqsubset$. 
Indeed, \cite{SL97} shows that all the immediate successors of an mvs $T$ can
be obtained as $T^u$, for some $u \in T$.

As an example, let us consider again the graph in Fig.~\ref{fig:nobcritical} and 
let $T$ be the mvs $\{a,b,c\}$. We have that $T^a=T^b=\{u,c\} \sqsubset T^c=\{t\}$.  

A $b$-minimal predecessor (w.r.t. $\sqsubset$) of $T$ is computed by function 
\fun{minimalMvs} in Alg.~\ref{alg:mmvs}. 
Given an mvs $T$ and a node $u\in T$, $T_{u} = (T\cup\fun{pred}(u))\setminus{\cal I}_s(T\cup\fun{pred}(u))$ 
is the analogue
of $T^u$ in the backward direction, where $\fun{pred}(u) = \{v \in V : v \rightarrow u \in E\}$ is the set
of immediate predecessors of $u$.  
The correctness of function 
\fun{minimalMvs} is given by the following result.

\begin{theorem}
\label{thm:bminimalChar}
$T$ is $b$-minimal if and only if $(T \cup \fun{pred}(u)) \setminus \{b\} \preceq b$, for every $u \in T$.
\end{theorem}

\noindent
{\bf Analysis\ }
First of all, we can calculate the reachability relation for every pair of vertices
in ${\cal O}(|V|^3)$, to fill in a $|V|\times |V|$ binary matrix that allows us to
check whether $T$ is critical in ${\cal O}(|V|^2)$ in Alg.~\ref{alg:wonw2}.
Second, relations $T' \sqsubseteq T$ and $T \preceq T'$ 
(in \fun{immediateMvsRight} and \fun{minimalMvs})
can be calculated in ${\cal O}(|V|^2)$ by calculating the connected component $C$ of the sink (for $\sqsubseteq$)
or of the source (for $\preceq$) in the subgraph $G[V\setminus T]$; the desired relation 
holds if and only if $T' \cap C = \varnothing$.
Third, given $T$ and $u \in T$, the mvs $T^u$ 
can be calculated in ${\cal O}(|V|^2)$,
by following \cite{SL97}: first, calculate the connected component $C$ of the sink in the subgraph 
$G[V\setminus (T \cup \fun{succ}(u))]$; then, ${\cal I}_t(T\cup\fun{succ}(u))$ is the subset of
$T \cup \fun{succ}(u)$ without an immediate successor in $C$.
Finally, $T_u$ in Alg.~\ref{alg:mmvs} can be computed symmetrically, by considering the connected component of the source
in $G[V\setminus (T \cup \fun{pred}(u))]$ and by excluding all vertices without an immediate predecessor in $C$.

Thus, function \fun{immediateMvsRight} costs ${\cal O}(|V|^3)$
since it computes at most $|V|$ mvs's to find $\min_{u\in T}T^u$.
%
Function \fun{minimalMvs} in Alg.~\ref{alg:mmvs} costs ${\cal O}(|V|^4)$: 
we need ${\cal O}(|V|)$ iterations of the {\bf while} of line 2 and
each iteration costs ${\cal O}(|V|^3)$, since, for every $u \in T$, we have to check relation $\preceq$.

Finally, Alg.~\ref{alg:wonw2} costs ${\cal O}(|V|^5)$, since  
function \fun{minimalMvs} is invoked at most once for each node in $V$. 
Indeed, if $b$ appears as a new node in $T'\setminus T$ (line 6 of Alg.~\ref{alg:wonw2}),
it cannot have already appeared in a $T'' \sqsubset T'$ of the chain, 
otherwise there would exist a walk from $b$ to $b$ and hence
$T''$ (as well as $T'$) would be critical and consequently function \fun{weakOrNotWeak} 
would have terminated in line 5 returning {\sc Weak} as a result.

%

\vspace{-3mm}
\section{Inefficiency in Flow Networks: Edge-weakness}
\label{sec:edgeweak}
\vspace{-2mm}

Interestingly, the notion of weakness strongly depends on the fact that
we assign charges to nodes, and not capacities to edges. Indeed, 
in several settings the two models are interchangeable \cite{amo93}. 
However, the analogous of the notion of 
weakness in the standard flow network model with capacities on edges, 
that we call {\em edge-weakness}, does not correspond to weakness.  

We first briefly recall some standard notions of flow networks \cite{amo93}.
First, we denote with ${\sl out}(u)$ the set of all edges whose first
component is $u$ and ${\sl in}(u)$ the set of all edges whose second component is $u$.
A {\em flow network} is a graph $G = (V,E)$ endowed with a capacity 
function $\{c_e\}_{e \in E}$, assigning a non-negative number ($c_e\in\mathbb{R}^+$) to every edge.
%
A {\em flow} in such a network is a function $f: E \rightarrow \mathbb{R}^+$ such that
\vspace*{-.2cm}
\begin{itemize}
\item $\forall e \in E.\, 0 \leq f(e) \leq c_e$, and
\item $\forall u \in V \setminus\{s,t\}.\, \sum_{e \in {\sl in}(u)} f(e) = \sum_{e \in {\sl out}(u)} f(e)$.
\vspace*{-.2cm}
\end{itemize}
The {\em value} of a flow $f$, written $|f|$, is defined as $\sum_{e \in {\sl out}(s)} f(e)$ and it turns out
to be equal to $\sum_{e \in {\sl in}(t)} f(e)$.

A flow $f$ {\em saturates} a network if, for every path, 
there exists an edge $e$ belonging to that path such that $f(e) = c_e$.
The standard problem in flow networks is to find a saturating flow with maximum value. 
%
By mimicking Def.~\ref{def:weak}, we give the following definition of edge-weak graph:

\begin{definition}
\label{def:edgeweak}
A graph is {\em edge-weak} if there exists a capacity assignment to edges 
such that the resulting flow network admits a non-maximum saturating flow. 
\vspace{-1mm}
\end{definition}


In general, we can calculate a flow in a network by
non-deterministically choosing paths and saturating them. It is easy to see that
such an algorithm always calculates a maximum flow if and only if
the graph is not edge-weak. In this sense, edge-weakness can be
considered another form of inefficiency in network design: to calculate a maximum flow, we cannot
use a simple iterated DFS but we need more complex algorithms (see \cite{amo93}).
Something similar happens in \cite{BBT85}. In that model, edges are also equipped with 
a cost function and the problem is to find a flow of a given value (between 0 and the maximum) 
but with the lowest possible cost. The authors prove that a greedy algorithm always calculates 
the minimum cost flow if and only if the underlying graph is series-parallel. Incidentally, we
will prove that the notion of series-parallel and non-edge-weakness coincide for acyclic
graphs (see Theorem \ref{thm:vulweak} later on).

We now graph-theoretically characterize the notion of edge-weakness. To this aim,
recall that a cut of a graph $G$ is a bipartition of its vertices $(S,T)$ such that $s \in S$ and
$t \in T$; moreover, its {\em cut-set} is the set of edges $(u,v)$ with $u \in S$ and $v \in T$.
We call {\em connected} a cut $(S,T)$ where  
every node in $S$
can be reached from $s$ without touching nodes in $T$ and every node in $T$ can reach $t$
without touching nodes in $S$. 

\begin{theorem}
	\vspace{-1mm}
\label{thm:fragile}
$G$ is edge-weak if and only if there exists a walk passing at least twice through the cut-set 
of some connected cut.
	\vspace{-2mm}
\end{theorem}

The proof is similar to the proof of Theorem \ref{thm:weak-characterize}.
Just notice that working with connected cuts ensures a minimality property
on the associated cut-set: given a connected cut $(S,T)$, no cut has a cut-set properly
contained in the cut-set of $(S,T)$. This somehow corresponds to the minimality
condition underlying an mvs.

\vspace{-2mm}
\section{Inefficiency in Traffic Networks: Vulnerability}
\label{sec:vuln}
\vspace{-2mm}

{\em Traffic Networks} \cite{braessFormal,BI97} provide a model for studying selfish routing: 
non-cooperative agents travel from a source node $s$ to 
a destination node $t$. Since the cost (or latency) experienced by an agent 
while travelling along a path depends on network congestion
(and hence on routes chosen by other agents), 
traffic in a network stabilizes to the equilibrium of 
a non-cooperative game, where all agents experience the same latency. 
This phenomenon has been defined by Wardrop \cite{War52} in the
contest of transport analysis.

In the following, 
we essentially follow the presentation in~\cite{Rough06}.
%
Here, a \emph{flow} for a graph $G=(V, E)$ is a function 
$\varphi:P(G) \rightarrow \mathbb{R}^+$, where $P(G)$ is the set of paths in $G$.
A flow induces a unique flow on edges: for any edge $e\in E$, 
$\varphi(e)=\sum_{p\in P(G): e\in p}\varphi(p)$. 
Since we do not have capacities on
edges (
as in the standard flow networks) or charges on nodes
(as in Depletable Channels), in this model  
a flow is simply a function assigning non-negative reals to paths, without
any further constraint.

A {\em latency function} $l_e:\mathbb{R}^+\rightarrow \mathbb{R}^+$ 
assigns to each edge $e$ a latency that depends on the flow on it;
as usual, we only consider continuous and non-decreasing latency functions.
The latency of a path $p$ under a flow $\varphi$ is the sum of the latencies of 
all edges in the path under $\varphi$, i.e., $l_p(\varphi)=\sum_{e\in p} l_e(\varphi(e))$.
If $H$ is a subgraph of $G$, we denote with 
$l|_H$ the restriction of the latency function $l$ on the edges of $H$.

Given a graph $G$, a real number $r\in\mathbb{R}^+$ and a latency function $l$, 
we call the triple $(G, r, l)$ an {\em instance}. A flow $\varphi$ is {\em feasible} for 
$(G, r, l)$ if the value of $\varphi$ is $r$. Notice that, since we do not have any
constraint on edges or vertices, every $r$ admits at least one feasible flow.

A flow $\varphi$ feasible for $(G, r, l)$ is at {\em Wardrop equilibrium} 
(or is a {\em Wardrop flow}) if, for all pairs of paths $p, q\in P(G)$ 
such that $\varphi(p)>0$, we have $l_p(\varphi)\leq l_q(\varphi)$.
In particular, this implies that, if $\varphi$ is a Wardrop flow, 
all paths to which $\varphi$ assigns a positive flow
have the same latency. It is known \cite{Rough06} that every
instance admits a Wardrop flow and that different Wardrop flows for the same
instance have the same latency along the same path. Thus,
we denote with $L(G, r, l)$ the latency of all paths 
with positive flow at Wardrop equilibrium. In the special case where $r = 0$,
we let $L(G, r, l)$ be 0.

{\em Braess's paradox} \cite{braessFormal,braessOriginal} originates when latency 
at Wardrop equilibrium decreases because of removing edges (or equivalently, by raising 
the latency function on edges): an instance $(G, r, l)$ suffers from Braess's paradox 
if there is a subgraph of $G$ with a lower latency.
Fig.~\ref{fig:wheatstone} shows the {\em Wheatstone network}, a minimal example of Braess's paradox.
A Wardrop flow of value 1 assigns all the flow to the path $s\,u\,v\,t$ in the picture. 
The latency in such a case is 2.   
In Fig.~\ref{fig:optimalSubgraph}, we show the optimal subgraph: in this case, 
a Wardrop flow of value 1 assigns 
\ifdan $\frac{1}{2}$ \else $\sfrac{1}{2}$ \fi
to both paths in the network, 
thus obtaining a latency of 
\ifdan $\frac{3}{2}$. \else $\sfrac{3}{2}$.\fi
%

	\vspace{-2mm}
\begin{definition}
	A graph $G$ is {\em vulnerable} if there exist a value $r$, a latency function
	$l$ and a subgraph $H$ of $G$ such that $L(G, r, l)>L(H, r, l|_H)$.
	\vspace{-1mm}
\end{definition}

\begin{figure}[t]
\begin{tabular}{cc}
\begin{minipage}{0.46\textwidth}
\center
$
\xymatrix@R=12pt@C=18pt{
   & u \ar[rd]^1\ar[dd]^0\\
s \ar[ur]^x \ar[dr]_1 && t\\
   & v \ar[ru]_x
   }
$
\caption{The Wheatstone network}
\label{fig:wheatstone}   
\end{minipage}
~
&
~
\begin{minipage}{0.46\textwidth}
\center
$
\xymatrix@R=12pt@C=18pt{
   & u \ar[rd]^1\\
s \ar[ur]^x \ar[dr]_1 && t\\
   & v \ar[ru]_x
   }
$   
\caption{Wheatstone optimal subgraph}
\label{fig:optimalSubgraph}   
\end{minipage}
\end{tabular}
\vspace{-3mm}
\end{figure}

A characterization of vulnerable undirected graphs is presented in~\cite{Milch06}. 
In particular, in~\cite{Milch06} it is proved that an undirected graph is vulnerable if and only if it contains (the undirected version of) the Wheatstone network.
We now show that the same result holds also for directed graphs, 
thus answering to Open Question 1 in Sect. 6.1 of~\cite{Rough06}.
To formally state our result, recall that \cite{VTL82}
$G$ contains a subgraph homeomorphic to $H$ 
if $H$ can be obtained from $G$ by a sequence of the following operations:
\vspace*{-.3cm}
\begin{itemize}
\item remove an edge;
\item replace $(u,v)$ and $(v,w)$ with $(u,w)$ and delete vertex $v$, 
whenever $(u,v)$ is the only edge entering into $v$ 
and $(v,w)$ is the only edge leaving $v$.
\end{itemize}
\vspace*{-.2cm}
Let us call $W$ the graph in Fig.\ref{fig:example}(b), i.e. the graph underlying the Wheatstone network.

\vspace{-2mm}
\begin{theorem}
\label{thm:vuln}
$G$ is vulnerable if and only if it contains a subgraph 
homeomorphic to $W$. 
\vspace{-2mm}
\end{theorem}

One direction can be proved by assigning a big-enough latency to all edges
that do not belong to the homeomorphic copy of $W$ and by mimicking
the latencies in Fig. \ref{fig:wheatstone} for the remaining ones. 
For the converse, we exploit the fact that, for acyclic
graphs, vulnerability coincides with not being series-parallel; this is 
Theorem 1 from \cite{ChenEtal15}\footnote{
Indeed, acyclic graphs are a specific case of what they call
{\em irredundant} graphs. A graph is said to be irredundant is every edge and every node
belongs to a simple (i.e., acyclic) $st$-path.}
and yields the result, by using \cite{Duf65}
(where it is proved that, for undirected graphs, being series-parallel coincides with
not having a subgraph homeomorphic to $W$; this result scales only to {\em acyclic} directed graphs).
We then reduce the cyclic case to the acyclic one by showing that,
in every net that does not contain $W$, we can remove cycles without 
changing the set of acyclic paths and by proving that this operation does not affect vulnerability.
Thus, we cannot claim that every non-vulnerable graph is series-parallel
but only that it cannot contain a subgraph homeomorphic to $W$.
An example of a graph that is not vulnerable and not series-parallel will be given in the next section
(and, of course, it is cyclic).

To conclude, we should mention \cite{matroid15}, where
an elegant generalization to all congestion games is given and a characterization of 
structures that do not suffer of Braess's paradox is given in terms of matroids. 
Differently from our result, this elegant characterization is not directly related to graph theoretic concepts.

\vspace{-2mm}
\section{Comparing Weakness, Edge-Weakness, and Vulnerability}
\label{sec:compare}
\vspace{-1mm}

Stemming from the characterizations of weakness (Theorem~\ref{thm:weak-characterize}),
edge-weakness (Theorem \ref{thm:fragile}) and vulnerability (Theorem~\ref{thm:vuln}), we can relate
the three notions of inefficiency studied so far. The precise picture (for general directed graphs) is
given in the top-left part of Fig.~\ref{fig:incl}.

We first show that vulnerability implies edge-weakness and that it can be characterized
by containment of an acyclic weak subgraph.

\begin{theorem}
\label{thm:vuln-edgew}
If $G$ is vulnerable,  then it is edge-weak.
\end{theorem}

\begin{theorem}
\label{thm:vuln-weak}
$G$ is vulnerable if and only if it contains a weak acyclic subgraph.
\end{theorem}

It is easy to see that the graph $W$ 
is weak, vulnerable and edge-weak.
Let us now consider the other graphs in Fig.~\ref{fig:incl}. 
The graph $A$ is weak, because its mvs $\{u\}$ is critical. 
By contrast, it is not 
edge-weak (and so neither vulnerable): its only connected cuts are $(\{s\},\{u,v,t\})$ and $(\{s,u,v\},\{t\})$ and no walk passes through
their cut-sets twice.
The graph $B$ is weak (its mvs $\{u\}$ is critical) and edge-weak
(there is a walk passing twice through the cut-set of its connected cut $(\{s,u\},\{v,t\})$).
By contrast, it is not vulnerable, since it does not contain a subgraph homeomorphic to $W$.
The graph $C$ is not weak: 
its mvs's are $\{s\}$, $\{t\}$, $\{v_1, v_2\}$ and $\{v_3, v_4\}$ 
and they are not critical.
By contrast, this graph is vulnerable 
because it trivially contains a subgraph homeomorphic to $W$.
%
Finally, let us consider the graph $D$: it is not weak 
(its mvs's are $\{s\}$, $\{t\}$, $\{u\}$ and $\{v\}$, which are not critical) nor vulnerable
(it does not contain a subgraph homeomorphic to $W$). By contrast, it is edge-weak, because
there is a walk that passes twice through the cut-set of the connected cut $(\{s,u,x\},\{y,v,t\})$.

\begin{figure}[t]
\begin{center}
\begin{tabular}{c}
\begin{tikzpicture}
                        \draw [color=red] (1,1) ellipse (2 and 1);
                        \draw [color= blue] (3,1) ellipse (1 and .5);
                        \draw [color=olive] (3,1) ellipse (2 and 1);
\draw [color=red] (1,0) node[below left] {Weak};
\draw [color=olive] (3,0) node[below right] {Edge-Weak};
\draw [color=blue] (2.3,1.5) node[above right] {Vulnerable};
\draw (0,1) node {$A$};
\draw (1.5,1) node {$B$};
\draw (2.5,1) node {$W$};
\draw (3.5,1) node {$C$};
\draw (4.5,1) node {$D$};
\end{tikzpicture}
\qquad\qquad
\begin{tikzpicture}
                        \draw [color=red] (1,1) circle (.8);
                        \draw [color= blue] (1.5,1) ellipse (1.5 and 1);
                        \draw [color=olive] (1.5,1) ellipse (1.44 and 0.94);
\draw [color=red] (1.5,.8) node[below left] {Weak};
\draw [color=olive] (3,1.5) node[below right] {Edge-Weak};
\draw [color=blue] (2.6,1.5) node[above right] {Vulnerable};
\draw (4.3,1.5) node[above right] {=};
\draw (1.2,1.2) node {$W$};
\draw (2.2,1.2) node {$C$};
\end{tikzpicture}
\vspace*{-.5cm}
\\
\qquad
\\
$
\xymatrix@R=5pt@C=10pt{
&\\
s \ar[r] & u 
\ar@/^/[dd]
\ar[r] & t\\
& & \\
& v \ar@/^/[uu] & \\
& A
}
$
\qquad
$
\xymatrix@R=6pt@C=10pt{
&\\
s \ar[r] & u \ar@/^/[r]
& v \ar@/^/[l] \ar[r] & t\\
&\\
&\\
& B
}
$
\qquad
$
\xymatrix@R=7pt@C=12pt{
   & v_1 \ar[rdd] \ar[r] & v_3 \ar[rd]\\
s \ar[ur] \ar[dr] &&& t\\
   & v_2 \ar[r] \ar[uur] & v_4 \ar[ru]
\\
& C
}
$
\qquad
$
\xymatrix@R=8pt@C=10pt{
&\\
  s \ar[r] & u\ar[r]\ar[d] & v\ar[r] & t\\
               & x \ar@/_/[r] & y \ar@/_/[l]\ar[u]
\\
& D
}
\vspace*{-.2cm}
$
\end{tabular}
\caption{The inclusion diagram for cyclic (top-left) and acyclic (top-right) directed graphs
($W$ is the Wheatstone graph -- see Fig. \ref{fig:example}(b))}
\vspace*{-.5cm}
\label{fig:incl}
\end{center}
\end{figure}
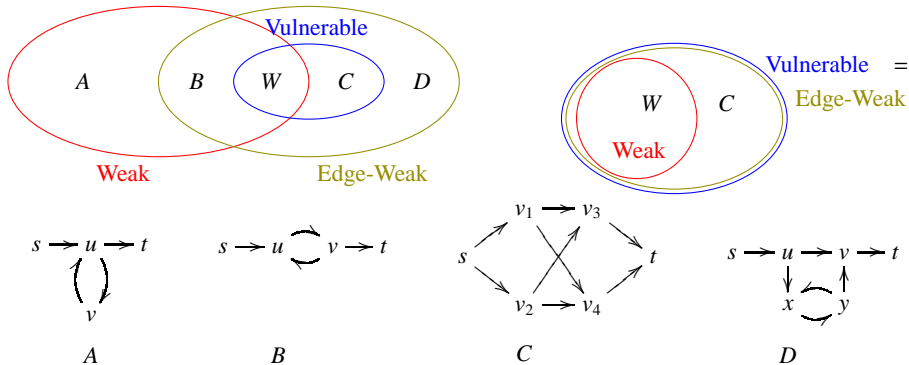


Graph $A$ also testifies that the characterization of vulnerability
given in \cite{ChenEtal15} does hold only for irredundant graphs. 
Indeed, Lemma 2 in  \cite{ChenEtal15} does not hold for $A$ (that is redundant because of the vertex $v$): 
the graph is cyclic but it does not contain what they call an $s$-$t$ paradox (see Def. 5 in  \cite{ChenEtal15}).

If we restrict ourselves to {\em acyclic} graphs (i.e., DAGs), the inclusion diagram changes;
it is depicted in the top-right part of Fig. \ref{fig:incl}. Indeed, for DAGs, we can prove that
both vulnerability and edge-weakness coincide with not being {\em two-terminal series-parallel} 
(TTSP \cite{RS42}); furthermore, the resulting class (properly) contains weak DAGs.

\begin{theorem}
\label{thm:vulweak}
If $G$ is a DAG, the following statements are equivalent:\\
\hspace*{.8cm} 
1. $G$ is vulnerable;
\hspace*{3cm}
3. $G$ is edge-weak;\\
\hspace*{.8cm} 
2. $G$ is not TTSP;
\hspace*{3.15cm} 
4. $G$ contains a weak subgraph.
\end{theorem}


	\vspace{-2mm}
\section{Conclusion}
\label{sec:conclu}
	\vspace{-1mm}

We have studied different models for networks: one is provided with
depletable node charge modeling energy consumption in ad-hoc networks;
one is the standard flow network model;
and one is provided with latency functions on the edges in a game theoretic framework 
for modeling traffic networks. 
In all models, a graph-theoretical notion of inefficiency can be identified, 
always related to a counterintuitive phenomenon: 
some networks increase their performances when an edge is removed. 
We have compared these forms of inefficiency and we have found precise 
relationships.

We have also shown a polynomial time algorithm for checking weakness.
Our algorithm may seem overly complicated. Indeed, because of Theorem\,\ref{thm:bminimal},
it would be enough to compute, for every $b \in V$, a $b$-minimal mvs and check whether it is
critical or not. The problem is that not every vertex occurs in an mvs (see, for example, vertex $v$
of graph $A$ in Fig. \ref{fig:incl}). Moreover, for a vertex $b$ occurring in an mvs, we have not
found an efficient way to directly compute a $b$-minimal mvs; indeed, Alg.~\ref{alg:mmvs} computes
a $b$-minimal mvs starting from an mvs that already contains $b$. Efficiently finding such an mvs
is left for future development.

The characterization we gave for vulnerability gives also hints on its computational complexity: 
it is polynomial for general graphs (a $O(|V|^5)$ algorithm can be easily extracted from the proof
of Theorem \ref{thm:vuln}). 
Finally, edge-weakness is polynomial in the acyclic case (since it coincides with vulnerability)
but we still do not know in the general case; possibly, by following the ideas underlying
the algorithm for weakness, a polynomial-time algorithm for edge-weakness may be devised.

\medskip\noindent
{\bf Acknowledgements\ }
The authors wish to thank Irene Finocchi and Fabrizio Grandoni 
for helpful discussions about the topic of this paper.  

\bibliographystyle{abbrv}
\bibliography{threeForms}

\newpage
\appendix
\section{Proofs}

\subsection{Proofs of Section \ref{sec:weak}}

{\bf Proof of Theorem~\ref{thm:weak-characterize}}
\\
\noindent
{\it (If)} We define $\eta(v)=1$ for all $v\in T$ and we let $\eta(v)$ be sufficiently
large on all other vertices, so that $\ma_\eta=|T|$. We may assume without loss of
generality that the walk $r:a\leadsto b$ is acyclic and such that $r\,\cap T=\{a, b\}$.
Similarly, by the minimality of $T$, there exist (acyclic) directed walks $p:s\leadsto a$
and $q:b\leadsto t$ such that $p\,\cap T=\{a\}$ and $q\,\cap T=\{b\}$. Then, a flow $\phi$ of one
unit along the path $s\leadsto a\leadsto b\leadsto t$ is feasible and it leaves the channel
with a charge-to-node assignment $\theta$ such that $\ma_\theta\leq |T|-2$.
Then, we can combine $\phi$ with a maximum flow for $\theta$ to obtain a dead network after 
a flow of value $\ma_\theta+1<\ma_\eta$. Thus, $G$ is weak.

{\it (Only if)} Let $\eta$ be a channel inhibited by some flow $\phi$ 
of value $n<\ma_\eta$ and call $\zeta$ the resulting (dead) channel. There exists 
an mvs $T$ in the graph such that $\zeta(v)=0$, for all $v\in T$.
Since every flow cannot exceed the capacity of an mvs,\footnote{
	This comes from the min-cut-max-flow theorem \cite{amo93}, 
	that can be easily rephrased in our setting to sound as {\em min-mvs-max-flow}. Indeed,
	we can adopt the standard translation from vertex-capacitated nets into edge-capacitated nets:
	we replace every vertex $v$ different from $s$ and $t$ with two new vertices $v_i$ and $v_o$
	and we add a new edge $v_i \rightarrow v_o$ whose capacity is the charge of node $v$;
	every edge $u \rightarrow v$ is replaced with the edge $u_o \rightarrow v_i$ whose capacity is $\infty$;
	every edge $s \rightarrow v$ is replaced with the edge $s \rightarrow v_i$ whose capacity is $\infty$;
	every edge $u \rightarrow t$ is replaced with the edge $u_o \rightarrow t$ whose capacity is $\infty$.
	Now, there is a one-to-one correspondence between the mvs's of the vertex-capacitated model
	and the cuts of the corresponding edge-capacitated model whose cut-set is formed 
	only by edges  $v_i \rightarrow v_o$; moreover, the capacity of any mvs and of the corresponding cut
	coincide. Since only cuts arising from mvs's have a finite capacity, it is easy to show that
	a min-cut comes from a min-mvs and, conversely, that a min-mvs induces a min-cut.
	Moreover, since only edges of the form $v_i \rightarrow v_o$ put constraints 
	on a flow, a max-flow in the edge-capacitated model
	corresponds to a max-flow in the vertex-capacitated model, and vice versa.
} 
$\ma_\eta\leq\sum_{v:v\in T}\eta(v)$.
By definition $\zeta(v)=\eta(v)-\sum_{p:v\in p}\phi(p)$ and hence $\eta(v)=\sum_{p:v\in p}\phi(p)$.
Suppose no directed walk exists between any two vertices of $T$. Thus, since $T$ is an mvs, all paths
must include \emph{precisely} one vertex in $T$;
hence, $\sum_{v:v\in T}\sum_{p:v\in p}\phi(p)=\sum_{p}\phi(p)$. Summing up, we have the absurd:
\[
\vspace{-.6cm}
\ma_\eta\leq\sum_{v:v\in T}\eta(v) = \sum_{v:v\in T}\sum_{p:v\in p}\phi(p)=\sum_{p}\phi(p)=n.
\]
\qed

\subsection{Proofs of Section \ref{sec:complexity}}

\noindent
{\bf Proof of Theorem~\ref{thm:criticalNodes}\ }
Looking for a contradiction, let us suppose that in a graph there is 
a critical mvs $T$ with a return in $b\in T$, but there exists a complete chain 
$T_0, T_1, \ldots, T_n$ such that for all $i$, $T_i$ does not contain any critical node. 

Trivially, $b\not\sqsubseteq T_0$ and $b\sqsubseteq T_n$ (this happens for every $b \in V\setminus \{s\}$); thus,
let us consider the index $i$ such that $b\not\sqsubseteq T_i$ and $b\sqsubseteq T_{i+1}$. 
By construction, $T_{i+1}= T_i^{u_i}$, for some $u_i \in T_i$;
moreover, $b\not\in T_{i+1}$, because, by hypothesis, no critical node belongs to any mvs in the chain.

Let us define the many-steps predecessors of $b$ as $\fun{pred}^\ast(b) = \{v \in V : v \leadsto b\}$ and take $P=\fun{pred}^\ast(b)\cap T_i$, the set of many-steps predecessors 
of $b$ in $T_i$. For this set, we observe two things:
\begin{description}
\item[(i)] $\varnothing \neq P \preceq b$: since $T_i$ is an mvs, every path passing through $b$ must cross $T_i$ but,
since $b\not\sqsubseteq T_i$, it must be that $T_i \preceq b$. Trivially, since all nodes between $s$
and $b$ are by definition the many-step predecessors of $b$, the same relation holds by restricting $T_i$ to 
such nodes, thus obtaining $P$. Since $b \neq s$, this entails that $P \neq \varnothing$.

\item[(ii)] For all $p\in P$, $T^p_i = T_{i+1}$: 
If some $p\in P$ belonged to $T_{i+1}$, we could find a walk $p\leadsto b\leadsto u$ with $u\in T_{i+1}$
(indeed, $b\sqsubseteq T_{i+1}$); this would make $T_{i+1}$ a critical mvs. 
So, we must have that all nodes in $P$ disappear in $T_{i+1}$; by construction, this
happens because, for every $p \in P$, we have that $p \in {\cal I}_t(T_i\cup\fun{succ}(u_i))$
and hence $p \sqsubseteq T_{i+1}$. This fact, together with $p \not\in T_{i+1}$,
entails that $T^p_i \sqsubseteq T_{i+1}$, for all $p\in P$.
If $T^p_i\not=T_{i+1}$ for some $p\in P$, we contradict the hypothesis that $T_0, T_1, \ldots, T_n$ 
is a complete chain: indeed, we would have that $T_i\sqsubset T_p\sqsubset T_{i+1}$. 
\end{description}

Because of point $\bf (ii)$, the set $U$ of new nodes added from $T_i$ to $T_{i+1}$ is the same
whenever we refine $T_i$ with $u_i$ or any other $p \in P$; in particular, every node in $U$ is
an immediate successor of every node in $P$.
Moreover, we can claim the following fact about $U$:
\begin{description}
\item[(iii)] $b \sqsubseteq U$: since $b \sqsubseteq T_{i+1}$, every walk form $b$ to $t$ must cross
$T_{i+1}$. If it passes through a node $x \in T_i \setminus T_{i+1}$, then the walk $p \leadsto b \leadsto x$,
for any $p \in P$, would make $T_i$ critical.
\end{description}

We now use these facts to contradict the assumption that $b$ is a critical node.
Consider all the paths of the form $s\leadsto p\rightarrow u\leadsto t$,
with $p\in P$ and $u\in U$. Every mvs, to cut such paths, 
must contain either a set of nodes $P'\preceq P$
or a set of nodes $U' \sqsupseteq U$. 
In both cases, since $P' \preceq b$ and $b\sqsubseteq U'$ (because of $\bf (i)$ and $\bf (iii)$), 
$b$ cannot belong to any mvs (see Lemma~\ref{fact:one}) and hence cannot be a critical node.
\qed
\bigskip

\noindent
{\bf Proof of Theorem~\ref{thm:bminimal}\ }
Let $T$ be a $b$-critical mvs, where $a \in T$ is such that $a\leadsto b$.
By contradiction, assume the existence of an mvs $T^\ast$ that is both $b$-minimal and not $b$-critical. 
We shall now prove that this will entail that $a$ and $b$ cannot both belong together to the same mvs, 
thus contradicting the existence of $T$. 

Clearly, $a \not\in T^\ast$, because $T^\ast$ is not $b$-critical.
Moreover, 
the existence of a walk from the source $s$ to $a$ that passes through $T^\ast$ would 
imply the existence of a walk starting in $T^\ast$ that reaches $a$ and then $b$, 
by contradicting that $T^\ast$ is not $b$-critical.
Therefore, it must be $a \sqsubseteq T^\ast$.


Let $A=\{u\in T^\ast~|~ a\leadsto u\not=b\}=\{a_1,\ldots, a_n\}$. 
Notice that $A \neq \varnothing$, otherwise 
we would have $a\sqsubseteq b$ and therefore $a$ and $b$ could not belong to the 
same mvs (see Lemma~\ref{fact:one}).
For each element $a_i\in A$, let us now consider the mvs $T^\ast_{a_i}\sqsubseteq T^\ast$. 
By, $b$-minimality of $T^\ast$, $b$ does not belong to $T^\ast_{a_i}$. 
This implies that there exists a set of nodes $P_i\subseteq \fun{pred}(a_i)$ such that 
$P_i\preceq b$. Observe that there is some path from $s$ to $b$ that do not pass through $a$, 
otherwise $a\preceq b$ and then they could not belong to the same mvs.
Therefore, for each $i\in\{1,\ldots, n\}$, we can consider the set of nodes 
$B_i\subseteq P_i$ that are on a path from $s$ to $b$ that does not touch $a$.
Notice that, for all $i$, we have that $B_i \neq \varnothing$, otherwise all paths from 
$s$ to $b$ would pass through $a$.

Let us now fix an $i\in\{1,\ldots, n\}$, let us come back to $T$ and consider 
how it can cut all paths of the form $s\leadsto b_i\rightarrow a_i\leadsto t$, for all $b_i\in B_i$. 
 
If in $T$ there exists a set of nodes $L_i\preceq B_i$, this would imply that $b$ cannot be in $T$, because
$L_i \cup \{a\} \preceq B_i \cup \{a\} \preceq b$ and $L_i \cup \{a\} \subseteq T$ (see Lemma~\ref{fact:one}). 

So, it must be that at least one path $s \leadsto b_i \rightarrow a_i$ does 
not pass through $T$.
Therefore there must exist a set in $R_i\subseteq T$ such that $a_i \sqsubseteq R_i$.
Since this argument works for every $i$, we can consider $R = \bigcup_i R_i \subseteq T$.
Then, $A = 
\{a_1,\ldots, a_n\} \sqsubseteq \bigcup_i R_i = R$.
But again, since $a\sqsubseteq A\cup\{b\}\sqsubseteq R\cup\{b\}$ and $R\cup\{b\} \subseteq T$,
Lemma~\ref{fact:one} would imply that $a$ cannot belong to $T$.  
This is a contradiction with the initial choice of the mvs $T$ that contains both $a$ and $b$.  
\qed
\bigskip

\noindent
{\bf Proof of Theorem~\ref{thm:bminimalChar}\ }
Given an mvs $T$ and two nodes $u,b\in T$, let us consider the mvs $T_u \sqsubset T$. 
We have that $b \not\in T_u$ if and only if $b \in {\cal I}_s(T\cup\fun{pred}(u))$,
i.e. $(T\cup\fun{pred}(u)) \setminus \{b\} \prec b$. 
Since every predecessor (w.r.t. $\sqsubset$) of $T$ can be obtained as $T_u$, for
some $u \in T$, we easily conclude.
\qed

\subsection{Proofs of Section \ref{sec:edgeweak}}

Let us denote with $cutset(S,T)$ the cut-set of the cut $(S,T)$.

\begin{lemma}
\label{lemma:connectedcut}
Let $(S,T)$ be a connected cut. For every cut $(S',T')$, it holds that 
$cutset(S',T') \not\subset cutset(S,T)$.
\end{lemma}
\begin{proof}
By contradiction, assume a cut $(S',T')$ such that 
$cutset(S',T') \subset cutset(S,T)$ and let $(u,v) \in cutset(S,T) \setminus cutset(S',T')$.
Because $(S',T')$ is a cut and $(u,v) \not\in cutset(S',T')$, it can either be $\{u,v\} \subseteq S'$
or $\{u,v\} \subseteq T'$. 

In the first case, since $(S,T)$ is a connected cut, we know that there
exists a walk $v \leadsto t$ containing only vertices of $T$. But $v \in S'$ whereas $t \in T'$;
hence, there must exist a $(x,y) \in v \leadsto t$ such that $x \in S'$ and $y \in T'$.
Then, $(x,y) \in cutset(S',T')$, whereas $(x,y) \not\in cutset(S,T)$, because
$\{x,y\} \subseteq T$. This contradicts the assumption $cutset(S',T') \subset cutset(S,T)$.

In the second case, we work in a similar way, but consider the walk $s \leadsto u$ containing
only vertices of $S$.
\qed
\end{proof}

\noindent
{\bf Proof of Theorem~\ref{thm:fragile}}
Given a set of edges $X$, we write $c(X)$ to denote $\sum_{x\in X} c(x)$.

\medskip\noindent
{\it (If)} 
Let $(S,T)$ be the connected cut and $E'$ be its cut-set.
Let $p$ be a walk that passes through $E'$ at least twice, with $(u,v)$ and
$(x,y)$ be the first and the last edge of $E'$ touched by it. 
We define $c_e$ as the number of occurrences of $e$ in $p$, for all $e \in E'$, 
and we let $c_e$ be $c(E')+1$ on all other edges. 

First of all, notice that the maximum flow has value $c(E')$. This follows from the min-cut-max-flow
theorem \cite{amo93}, by noting that $(S,T)$ is a minimum cut. Indeed, by Lemma \ref{lemma:connectedcut},
no cut has a cut-set contained in $E'$; moreover, by definition of $c$, any cut whose cut-set contains
an edge not belonging to $E'$ has a capacity greater than $c(E')$.

By definition of cut-set,
$u \in S$ and $y \in T$; moreover, by construction, 
$s \leadsto u$ and $y\leadsto t$ do not pass through $E'$. Then, a flow $f$ of one
unit along $p$ is feasible and it leaves the net
with a residual capacity lower than $c(E')-1$.
Then, we can combine $f$ with a maximum flow for the residual net to obtain a saturating flow
of value at most $c(E')-1$. Thus, $G$ is edge-weak.

\medskip\noindent
{\it (Only If)} 
Let $c$ be a capacity function that admits a saturating flow $f$ 
of value smaller than the maximum and call $c'$ the resulting residual capacity. 
Since $f$ saturates $c$, we now show that there exists a connected cut whose cut-set has
a residual capacity that equals 0. If this was not the case, let us reason as follow.
Start with the cut $(\{s\}, V \setminus \{s\})$. Clearly, this is a connected cut; so, there exists
an edge $e_1$ from $s$ to some $u_1 \in V \setminus \{s\}$ such that $c'_{e_1} > 0$.
If $u_1 = t$, we have a contradiction with the fact that $f$ saturates $c$.
So, consider the cut $(\{s,u_1\}, V \setminus \{s,u_1\})$. Again, this is a connected cut and
so there exists an edge $e_2$ from $\{s,u_1\}$ to some $u_2 \in V \setminus \{s,u_1\}$ 
such that $c'_{e_2} > 0$. If $u_2 = t$, we have a contradiction with the fact that $f$ saturates $c$.
Otherwise, we go on, until we find a $u_k = t$. Then, the path $s \rightarrow u_1 \rightarrow u_2
\rightarrow \ldots \rightarrow u_k = t$ contradicts the fact that $f$ saturates $c$.

To conclude, let $(S,T)$ be the connected cut whose cut-set has been saturated by $f$.
By the min-cut-max-flow theorem \cite{amo93}, $c(E') \geq |f_{\rm max}|$. But this is possible
only if there is a path that passes through $E'$ at least twice. Indeed, if this was
not the case, we would have that $c(E') = |f|$, in contradiction with the assumption
$|f| < |f_{\rm \max}|$.
\qed

\subsection{Proofs of Section \ref{sec:vuln}}

First, we characterize vulnerability for acyclic graphs; this is an immediate corollary 
of the main result from \cite{ChenEtal15}: an irredundant graph is vulnerable
if and only if it is not series-parallel. 
Since acyclic graphs are a special case of irredundant ones, the theorem
needs no proof.

\begin{theorem}
\label{thm:vuln-ser-par}
Let $G$ be acyclic; then, $G$ is vulnerable if and only if it is not TTSP.
\end{theorem}

We are now ready to move to general (i.e., cyclic) graphs.
To this aim, let $G$ be a directed $st$-graph; 
we denote with ${\rm Path}_A(G)$ the set of acyclic paths from $s$ to $t$ in $G$. 

\begin{lemma}
\label{lemma:removeCycle}
Let $G$ be a directed $st$-graph that does not contain a subgraph homeomorphic to the Wheatstone network. 
Let $C={v_1,\ldots, v_n}$ be the nodes of a simple cycle in $G$. 
Then, we can remove an edge in $C$, obtaining a graph $G'$ such that 
${\rm Path}_A(G')={\rm Path}_A(G)$.
\end{lemma}
\begin{proof}
Let $(v_1,v_2),(v_2,v_3),\ldots,(v_{n-1},v_n),(v_n,v_1)$ be a simple cycle in $G$. 
Without loss of generality, let $v_1$ be a node such that there exists 
a walk $s\leadsto v_1$ that touches $C$ only in $v_1$. 
Such a node does exist, 
otherwise the cycle would not be reachable from the source $s$. 
Similarly, let $v_k$ be the last node in the sequence 
$v_{1}, \ldots, v_{n}$ such that there exists a walk 
$v_k\leadsto t$ that touches $C$ only in $v_k$. 
Such a node does exist, otherwise the sink $t$ would not be 
reachable from the cycle $C$. 

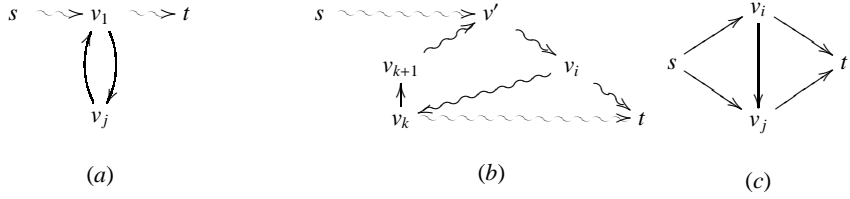
\begin{figure}[t]
\begin{tabular}{ccc}
\begin{minipage}{0.3\linewidth}
$
\xymatrix@R=10pt@C=20pt{
s \ar@{~>}[r] & v_1 
\ar@/^/[dd]
\ar@{~>}[r] & t\\
& & \\
& v_j \ar@/^/[uu] & \\
& (a)
\\
&
}
$
\end{minipage}
~
&
~
\begin{minipage}{0.35\linewidth}
$
\xymatrix@R=8pt@C=16pt{
s \ar@{~>}[rr]  & & v' \ar@{~>}[dr] & & \\
& v_{k+1} \ar@{~>}[ur]& &v_i\ar@{~>}[dll]\ar@{~>}[dr] & \\
& v_k \ar@{~>}[rrr]\ar@{->}[u] & & & t\\
&& (b)
}
$
\caption{Figures for Lemma~\ref{lemma:removeCycle}}
\label{fig:lemmacycle}
\end{minipage}
~
&
~
\begin{minipage}{0.3\linewidth}
$
\xymatrix@R=10pt@C=20pt{
   & v_i \ar[dd]\ar[dr]\\
s \ar[ur]\ar[dr] && t\\
   & v_j \ar[ur]\\
   & (c)
   \\
   &
   }
$
\end{minipage}
\end{tabular}
\end{figure}

If $k=1$, no acyclic path in $G$ passes through an edge in $C$ (see Fig.~\ref{fig:lemmacycle}(a)). 
In such a case, we can remove all edges in $C$, without removing any acyclic path in $G$.

Let us consider now a walk $v_1\leadsto t$. If all such walks have the form 
$v_1\leadsto v_k\leadsto t$, then 
$v_k$ is the only exit from $C$. Consequently, the edge $(v_k,v_{k+1\, {\rm mod} ~ n})$ 
does not belong to any acyclic path in $G$. Indeed all paths of the form 
$s \leadsto v_k \rightarrow v_{k+1\, {\rm mod} ~ n} \leadsto t$ must return in 
$v_k$ after $v_{k+1\, {\rm mod} ~ n}$ to leave the cycle $C$.

Otherwise, let $v_i$ ($1 \leq i < k$) be the first node in $v_{1}, \ldots, v_{k-1}$ 
that can reach $t$ without passing through other nodes in $C$
and consider nodes $v_{i+1}, \ldots, v_k$. If all walks from $s$ to them
pass through other nodes in $C$, again the 
edge $(v_k,v_{k+1\, {\rm mod} ~ n})$ 
does not belong to any acyclic path in $G$. 
Indeed, consider a path that uses edge $(v_k,v_{k+1\, {\rm mod} ~ n})$.
For what we have just assumed, such a path enters into the cycle $C$
in a node $v' \in \{v_{k+1\, {\rm mod}~ n}, \ldots, v_1, \ldots v_{i}\}$ and has to pass through edge 
$(v_k,v_{k+1\, {\rm mod}~ n})$. Since all vertices in $v_{k+1\, {\rm mod}~ n}, \ldots, v_1, \ldots v_{i}$
can reach $t$ only through $v_i$ or after it, the path is cyclic (see Fig.~\ref{fig:lemmacycle}(b)).

Otherwise, let $v_j$ ($i<j\leq k$) be such that there exists a walk $s\leadsto v_j$ that does not pass 
through other nodes in $C$. In this case, $G$ would contain a subgraph homeomorphic to the 
Wheatstone network, as given in Fig.~\ref{fig:lemmacycle}(c).
\qed
\end{proof}

\begin{lemma}
\label{lemma:equivalentNashFlows}
Let $G'\subseteq G$. Let $\varphi$ be an (acyclic) flow at the Wardrop equilibrium for $(G,r,l)$. 
If ${\rm Path}_A(G')={\rm Path}_A(G)$ then $\varphi$ is a flow at the Wardrop equilibrium 
for $(G',r,l)$.
\end{lemma}
\begin{proof}
By Proposition 2.2 in \cite{Rough06}, we know that, for every $p,q \in P(G)$, if $\varphi_p > 0$
then $\ell_p(\varphi) \leq \ell_q(\varphi)$. By hypothesis,  $\varphi$ assigns positive flow only to
acyclic paths in $G$; thus, $p \in {\rm Path}_A(G) = {\rm Path}_A(G')$. Moreover, since 
$G'\subseteq G$, it holds that $P(G') \subseteq P(G)$. Thus, trivially, for every $p,q \in P(G')$, 
if $\varphi_p > 0$ then $\ell_p(\varphi) \leq \ell_q(\varphi)$. Again by Proposition 2.2 in \cite{Rough06}, 
this means that $\varphi$ is a flow at the Wardrop equilibrium for $(G',r,l)$.
\end{proof}

\begin{lemma}
\label{lemma:equivalentGraphs}
Let $G'\subseteq G$.
If ${\rm Path}_A(G')={\rm Path}_A(G)$ then $G$ is vulnerable if and only if $G'$ is vulnerable. 
\end{lemma}
\begin{proof}
Vulnerability of $G'$ trivially entails vulnerability of $G$. Let us prove the opposite implication.
Let $H\subset G$ be such that $L(H,r,l)<L(G,r,l)$, for some $r$ and $l$. 

If $H \subset G'$, let $\varphi$ be an acyclic flow for $G$ at the Wardrop equilibrium
(one always exists by Proposition 2.4 of \cite{Rough06}). By Lemma~\ref{lemma:equivalentNashFlows},
$\varphi$ is a flow for $G'$ at the Wardrop equilibrium; thus, $L(H,r,l)<L(G',r,l)$, i.e. $G'$ is vulnerable.

Otherwise, it cannot be $H = G'$, because ${\rm Path}_A(H) \subset {\rm Path}_A(G)$; indeed,
because of Lemma~\ref{lemma:equivalentNashFlows}, if ${\rm Path}_A(H) = {\rm Path}_A(G)$,
we would have $L(H,r,l)=L(G,r,l)$. Thus, there is an edge $(u,v) \in H$ such that $(u,v) \not\in G'$;
this means that $(u,v)$ only belongs to cyclic paths of $G$, because by hypothesis
${\rm Path}_A(G')={\rm Path}_A(G)$. We now show that this implies the existence of a simple cycle
$p:w \leadsto w$ in $G$ containing $(u,v)$ such that: 
\begin{enumerate}
\item for every $w' \in w \leadsto u$, it holds that $w \sqsubseteq w'$, and
\item for every $w' \in v \leadsto w$, it holds that $w' \preceq w$. 
\end{enumerate}
First of all, since $(u,v)$ belongs to a cyclic path in $G$, it also belongs to a simple cycle in $G$, say $w \leadsto w$. 
To prove the first claim, consider $w' \neq w$ (the case for $w' = w$ is trivial). Then, if there was a path $s \leadsto w'$
not passing through $w$, then $s \leadsto w' \leadsto u \rightarrow v \leadsto w \leadsto t$ would be an acyclic path
in $G$, again by the assumption that $(u,v)$ only belongs to cyclic paths in $G$.
Similarly, to prove the second claim, assume a walk $w' \leadsto t$ not passing through $w$; 
then $s \leadsto w \leadsto u \rightarrow v \leadsto w' \leadsto t$ would be an acyclic path
in $G$, again the assumption that $(u,v)$ only belongs to cyclic paths in $G$.

Then, consider $H' = H \setminus p$. Trivially, ${\rm Path}_A(H')={\rm Path}_A(H)$ and
let $\varphi'$ be an acyclic flow for $H$ at the Wardrop equilibrium
(one always exists by Proposition 2.4 of \cite{Rough06}). By Lemma~\ref{lemma:equivalentNashFlows},
$\varphi'$ is a flow for $H'$ at the Wardrop equilibrium and $L(H,r,l)=L(H',r,l)$.
If $H' \subset G'$, then $L(H',r,l)<L(G',r,l)$ and $G'$ is vulnerable.
Otherwise, we can find another simple cycle to be removed from $H'$ but this
procedure has to terminate, eventually yielding that $G'$ is vulnerable, as desired.
\qed
\end{proof}

\noindent
{\bf Proof of Theorem \ref{thm:vuln}}
For the "if" part,
we know that $G$ admits a subgraph of the form
$$
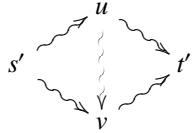
\xymatrix@R=10pt@C=20pt{
   && u \ar@{~>}[dd]\ar@{~>}[dr]\\
& s' \ar@{~>}[ur]\ar@{~>}[dr] && t'\\
   && v \ar@{~>}[ur]
   }
$$
Let us consider the latency assignment $l$ that assigns:
\begin{itemize}
\item 0 to all edges in $u \leadsto v$;
\item $x$ to the first edge in $s' \leadsto u$ and in $v \leadsto t'$ and 0 to all the remaining edges in those paths;
\item 1 to the first edge in $s' \leadsto v$ and in $u \leadsto t'$ and 0 to all the remaining edges in those paths;
\item $\infty$ to all the remaining edges. 
\end{itemize}
Trivially, this reproduces the Wheatstone network within $G$ that, consequently, is vulnerable.

For the "only if" part, 
if $G$ is not cyclic, the statement follows by Theorem~\ref{thm:vuln-ser-par}.
If $G$ contains cycles, by contradiction, let us suppose that it does not contain a subgraph 
homeomorphic to the Wheatstone network. 
Then, by Lemma~\ref{lemma:removeCycle}, we can transform $G$ into an acyclic subgraph $G'$ such that 
${\rm Path}_A(G')={\rm Path}_A(G)$. By Lemma \ref{lemma:equivalentGraphs}, since $G$ is vulnerable, $G'$ is vulnerable too; thus, since $G'$ is acyclic and not TTSP (by Theorem~\ref{thm:vuln-ser-par}), 
it contains a subgraph homeomorphic to the Wheatstone 
network \cite{Duf65} and, consequently, also $G$ does. Absurd.
\qed

\subsection{Proofs of Section \ref{sec:compare}}

\noindent
{\bf Proof of Theorem~\ref{thm:vuln-edgew}\ }
By Theorem \ref{thm:vuln}, $G$ admits a subgraph of the form
$$
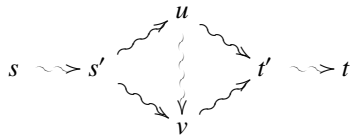
\xymatrix@R=10pt@C=20pt{
   && u \ar@{~>}[dd]\ar@{~>}[dr]\\
s \ar@{~>}[r] & s' \ar@{~>}[ur]\ar@{~>}[dr] && t' \ar@{~>}[r]& t\\
   && v \ar@{~>}[ur]
   }
$$
Let us consider the capacity assignment $\{c_e\}_{e \in E}$ that assigns
\begin{itemize}
\item 2 to all edges in $s\leadsto s'$ and $t'\leadsto t$;
\item 1 to all edges in $s'\leadsto u \cup s'\leadsto v 
\cup u\leadsto v \cup u\leadsto t' \cup v\leadsto t'$; and
\item 0 to all the remaining ones. 
\end{itemize}
Trivially, the flow assigning 1 to the path
$s\leadsto s' \leadsto u \leadsto v \leadsto t' \leadsto t$
is a non-maximum saturating flow, since there exists a flow with value 2.
\qed
\bigskip

\noindent
{\bf Proof of Theorem~\ref{thm:vuln-weak}\ }
If $G$ is vulnerable, then, by Theorem~\ref{thm:vuln}, it admits a subgraph $H$
of the form:
$$
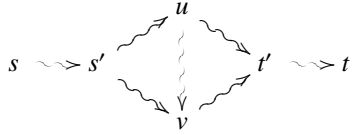
\xymatrix@R=10pt@C=20pt{
   && u \ar@{~>}[dd]\ar@{~>}[dr]\\
s \ar@{~>}[r] & s' \ar@{~>}[ur]\ar@{~>}[dr] && t' \ar@{~>}[r]& t\\
   && v \ar@{~>}[ur]
   }
$$
The subgraph $H$ is weak. As a matter of fact, $\{u,v\}$ is an mvs for $H$, 
and the walk 
$u\leadsto v$ allows us to conclude by Theorem~\ref{thm:weak-characterize}.


For the converse implication, let $H$ be a weak subgraph of $G$.
Theorem~\ref{thm:weak-characterize} implies that $H$ admits an mvs $T$
such that there exists a walk $a \leadsto b$ in $H$, with $\{a,b\} \subseteq T$.
By minimality of $T$, there exist in $H$ a walk from $s$ to $a$ (that does not 
contain $b$), a walk from $s$ to $b$ (that does not 
contain $a$), a walk from $a$ to $t$ (that does not 
contain $b$), and a walk from $b$ to $t$ (that does not 
contain $a$).  Thus, we have found the following subgraph of $H$ (and, hence, also of $G$):
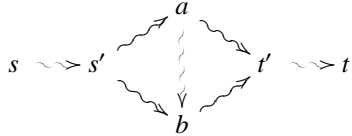
$$
\xymatrix@R=10pt@C=20pt{
   && a \ar@{~>}[dd]\ar@{~>}[dr]\\
s \ar@{~>}[r] & s' \ar@{~>}[ur]\ar@{~>}[dr] && t' \ar@{~>}[r]& t\\
   && b \ar@{~>}[ur]
   }
$$
Notice that both $s'$ and $t'$ cannot occur in $a \leadsto b$, otherwise we would
have a cycle in $H$, in contradiction with the assumption of its acyclicity. For the same
reason, $a \neq b$. Then, we can easily conclude that $G$ is vulnerable by Theorem~\ref{thm:vuln}.
\qed
\bigskip

\noindent
{\bf Proof of Theorem~\ref{thm:vulweak}\ }
\begin{description}
\item[(1) $\Rightarrow$ (2):] By Theorem \ref{thm:vuln-edgew}.
\item[(2) $\Rightarrow$ (3):] We prove the contrapositive and work by induction on the structure of $G$.

\quad
The base case is when $G$ is a single edge $(s,t)$. In this case, the only
saturating flow for $G$ is the one that saturates the capacity of $(s,t)$;
trivially, this is the maximum flow.

\quad
If $G$ is the serial composition of $G_1$ and $G_2$ (that are TTSP),
by the inductive hypothesis $G_i$ is not edge-weak, for $i \in \{1,2\}$. 
Now, since $E = E_1 \cup E_2$, every capacity
assignment $\{c_e\}_{e \in E}$ is $\{c_e\}_{e \in E_1} \cup \{c_e\}_{e \in E_2}$.
Moreover, because of serial composition, $\max(G,\{c_e\}_{e \in E}) = \min\{
\max(G_1,\{c_e\}_{e \in E_1}), \max(G_2,\{c_e\}_{e \in E_2})\}$. Finally,
since every source-to-sink path in $G$ is a source-to-sink path in $G_1$
followed by a source-to-sink path in $G_2$, every flow $f$ for $G$ induces
a flow with the same value $f_1$ for $G_1$ and a flow with the same value $f_2$ for $G_2$.
Now, by contradiction, assume that there exists $\{c_e\}_{e \in E}$ such that $(G,\{c_e\}_{e \in E})$ 
is edge-weak, i.e. it admits a saturating flow $\hat f$ such that $|\hat f| < \max(G,\{c_e\}_{e \in E})$. 
Thus, $\hat f_i$ saturates all the paths of $G_i$, for either $i = 1$ or $i = 2$.
But this is in contradiction with the inductive hypothesis, since
$|\hat f_i| < \max(G_i,\{c_e\}_{e \in E_i})$.

\quad
If $G$ is the parallel composition of $G_1$ and $G_2$ (that are TTSP),
the proof is similar. Just notice that in this case $\max(G,\{c_e\}_{e \in E}) = 
\max(G_1,\{c_e\}_{e \in E_1}) + \max(G_2,\{c_e\}_{e \in E_2})$, that every source-to-sink 
path in $G$ is a source-to-sink path either in $G_1$ or in $G_2$ and, consequently, that 
every flow $f$ for $G$ induces two flows, $f_1$ for $G_1$ and $f_2$ for $G_2$, such that
$|f| = |f_1| + |f_2|$. By contradiction, let $\{c_e\}_{e \in E}$ be such that $(G,\{c_e\}_{e \in E})$ 
admits a saturating flow $\hat f$ with $|\hat f| < \max(G,\{c_e\}_{e \in E})$. 
Then, $\hat f_i$ saturates all the paths of $G_i$, for $i \in\{1,2\}$.
The contradiction comes from the fact that $|\hat f_1| + |\hat f_2| <
\max(G_1,\{c_e\}_{e \in E_1}) + \max(G_2,\{c_e\}_{e \in E_2})$ and from
$|\hat f_i| \leq \max(G_i,\{c_e\}_{e \in E_i})$, for both $i \in \{1,2\}$.

\item[(3) $\Rightarrow$ (4):] If $G$ is a not TTSP DAG, then, by \cite{Duf65}, it admits a subgraph homeomorphic to
the Wheatstone graph; such a subgraph is weak (see the proof of Theorem \ref{thm:vuln-weak}).
\item[(4) $\Rightarrow$ (1):] By Theorem \ref{thm:vuln-weak}.
\qed
\end{description}
\bigskip

\end{document}